\documentclass[twocolumn, amssymb, nobibnotes, aps, pra, superscriptaddress, 10pt]{revtex4-2}
\usepackage{bbm}
\pdfoutput=1
\usepackage[utf8]{inputenc}
\usepackage[english]{babel}
\usepackage[T1]{fontenc}
\usepackage{bm}
\usepackage{mathtools}
\usepackage{amssymb}
\usepackage{amsthm}
\usepackage{physics}
\usepackage{braket}
\usepackage{nameref}
\usepackage{enumitem}
\usepackage[colorlinks=True, linkcolor= red, citecolor= blue]{hyperref}

\def\Tr#1{\text{tr}\left[{#1}\right]}
\def\tr#1{\text{tr}[{#1}]}
\def\Trr#1#2{\text{tr}_#1\left[{#2}\right]}
\def\Var#1{\text{var}\,{#1}}
\def\Cov#1{\text{cov}\left({#1}\right)}
\def\Ex#1#2{\mathbb{E}_{#1}\!\left({#2}\right)}
\def\Exb#1#2{\mathbb{E}_{#1}\!\left[{#2}\right]}
\def\Exc#1#2{\mathbb{E}_{#1}\!\left\{{#2}\right\}}
\def\ex#1{\mathbb{E}\,{#1}}
\def\Dir#1{D\left({#1}\right)}

\newcommand{\Id}{\mathbbm{1}}

\newcommand{\R}{\mathbb{R}}

\newcommand{\Hil}{\mathcal{H}}
\newcommand{\M}{\mathcal{M}}

\newtheorem{theorem}{Theorem}[section]
\newtheorem{corollary}{Corollary}[theorem]
\newtheorem{lemma}[theorem]{Lemma}

\begin{document}


\title{Estimating Bell diagonal states with separable measurements}

\author{Noah Kaufmann}
\email{noah.kaufmann@nbi.ku.dk}
\affiliation{Networked Quantum Devices Unit, Okinawa Institute of Science and Technology Graduate University, Okinawa, Japan}
\affiliation{Center for Hybrid Quantum Networks, Niels Bohr Institute, University of Copenhagen, Copenhagen, Denmark}

\author{Maria Quadeer}
\email{mariaquadeer@gmail.com}
\affiliation{Networked Quantum Devices Unit, Okinawa Institute of Science and Technology Graduate University, Okinawa, Japan}

\author{David Elkouss}
\email{david.elkouss@oist.jp}
\affiliation{Networked Quantum Devices Unit, Okinawa Institute of Science and Technology Graduate University, Okinawa, Japan}

\begin{abstract}
Quantum network protocols depend on the availability of shared entanglement. Given that entanglement generation and distribution are affected by noise, characterization of the shared entangled states is essential to bound the errors of the protocols. This work analyzes the estimation of Bell diagonal states within quantum networks, where operations are limited to local actions and classical communication. We demonstrate the advantages of Bayesian mean estimation over direct inversion and maximum-likelihood estimation, providing analytical expressions for estimation risk and supporting our findings with numerical simulations.
\end{abstract}

\maketitle


\section{Introduction}

In quantum information processing with a given set of operations, quantum states can be categorized as free states or resource states~\cite{Chitambar2019}. Free states are those that can be generated through the available set of operations, while resource states cannot. Access to resource states can greatly elevate a quantum device's power in computation and communication tasks, an example being magic states~\cite{Bravyi2005, gottesman1999, knill2005}. The interest of this work lies in entangled states, in particular Bell states, which are an essential resource in quantum networks. Bell states are at the center of many quantum network applications and are used, for example, in quantum key distribution~\cite{Ekert1991}, quantum teleportation~\cite{bennett1993}, and superdense coding~\cite{bennett1992}.

In quantum networks, where Bell states are subject to imperfect preparation and distribution, the error rate of a quantum information process depends on both the quality of the local operations and the fidelity of the shared resource states. While operations can be characterized using methods such as quantum process tomography~\cite{Chuang1997, Poyatos1997}, gate-set tomography~\cite{Merkel2013, blume-kohout2017}, or randomized benchmarking~\cite{Emerson2005, Knill2008, Magesan2011}, characterizing resource states involves some form of quantum state estimation.

A principal aim in designing state estimation protocols is to ensure that the estimation uncertainty scales efficiently with the number of measurements, particularly when characterizing resource states with limited distribution rates. This can be achieved by introducing a model for the imperfect states that focuses on the properties of interest and incorporates prior knowledge about the expected errors. A practical modeling approach for shared Bell pairs is to represent the noisy resource states as Bell diagonal states, which are statistical mixtures of Bell states. For example, a Bell state subjected to a depolarization channel results in such a Bell diagonal state~\cite{Liu_2017}. Bell diagonal states are also relevant in the security analysis of quantum key distribution protocols~\cite{Grasselli2021}. Moreover, since quantum communication protocols typically rely on knowing which of the four Bell states were initially shared, characterizing the diagonal elements of the density matrix in the Bell basis provides critical insights into the protocol's expected errors. Investigating the efficient characterization of Bell diagonal states is therefore a relevant problem.

After introducing concepts related to state estimation and Bell diagonal states, the main section of this work begins by establishing a general bound for the estimation of Bell diagonal states using quantum Fisher information. We show that this bound is achievable with Bell state measurements, which, though nonlocal, serve as a valuable reference for separable protocols. Next, we compare the analytical estimation risks of two approaches, i.e., Bayesian mean estimation (BME) and direct inversion, focusing on a network context where only local measurements and classical communication are allowed. The protocol relies on partial knowledge of the device and assumes trust in the channel messages, which can be established independently of the characterization protocol by a user authentication protocol~\cite{Dutta2022}. Finally, we present numerical results that examine the estimation risk across a broader range of measurement sets and estimators, highlighting the advantages of Bayesian mean estimation over both maximum-likelihood estimation and direct inversion.


\section{Preliminaries}

This section introduces key concepts and notation needed for the analysis and derivations that follow, focusing on Bell diagonal states, quantum state estimation techniques, and the quantum Cramér-Rao bound.

\subsection{Bell diagonal states}
The Bell states, also known as Einstein-Podolsky-Rosen pairs, form an orthonormal basis, called the Bell basis, for the two-qubit state space, where each state is maximally entangled~\cite{NielsenChuang2011}. The four Bell states are:
\begin{equation}
\label{eq:BellStates}
\begin{aligned}
    \ket{\Psi_1} &= \ket{\Phi^+} =  2^{-\frac{1}{2}}(\ket{00} + \ket{11})\\
    \ket{\Psi_2} &= \ket{\Phi^-} =  2^{-\frac{1}{2}} (\ket{00} - \ket{11})\\
    \ket{\Psi_3} &= \ket{\Psi^+} =  2^{-\frac{1}{2}} (\ket{01} + \ket{10})\\
    \ket{\Psi_4} &= \ket{\Psi^-} =  2^{-\frac{1}{2}} (\ket{01} - \ket{10}).
\end{aligned}
\end{equation}

The convex hull of the Bell states is the set of Bell diagonal states. They are a statistical mixture of Bell states and can be expressed in terms of the vector $\bm{\theta}$, with $\Sigma \theta_i = 1$ and $\theta_i \in [0, 1]$, as:
\begin{equation}
    \label{eq:spectral_decomposition}
    \rho = \sum_{i=1}^4 \theta_i \ketbra{\Psi_i}{\Psi_i}.
\end{equation}
In the following, we will use the notation $\Psi_i$ for the density matrix $\ketbra{\Psi_i}{\Psi_i}$.

\begin{figure}
  \centering
  \includegraphics[width=0.9\columnwidth]{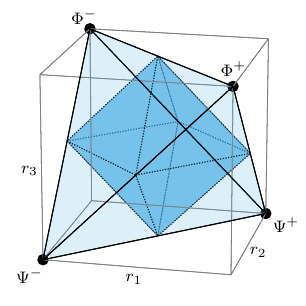}
  \caption{Bell diagonal states can be geometrically represented by interpreting vector $\bm{t}$ of Eq.~\eqref{eq:pauli_rep} as a point in $\R^3$ \cite{Horodecki1996}. All such states lie within a tetrahedron, formed as the convex hull of four vertices corresponding to the Bell states, $\Phi^+, \Phi^-, \Psi^+$, and $\Psi^-$. Within this tetrahedron, the set of separable states forms a shaded octahedron, satisfying the condition $\abs{t_1} + \abs{t_2} + \abs{t_3} \leq 1$. The completely mixed state is located at the center of the tetrahedron.}
  \label{fig:figure1}
\end{figure}

Another way to parametrize Bell diagonal states is by expressing them in a basis generated by the Pauli operators. Together with the identity, the Pauli matrices $\sigma_1 = \sigma_x, \sigma_2 = \sigma_y,$ and $\sigma_3 = \sigma_z$ form a basis $\mathbb{P}^{\otimes 2} = \{ \Id, \sigma_1, \sigma_2, \sigma_3 \}^{\otimes 2}$ for the $4 \cross 4$ Hermitian matrices. Since density matrices are Hermitian, a general two-qubit state $\rho$ can be represented in this Pauli basis by:
\begin{equation}
    \label{eq:gen_state_pauli_rep}
    \rho = \frac{1}{4}\Biggl( \Id \! \otimes \! \Id + \bm{a} \cdot \bm{\sigma} \! \otimes \! \Id + \Id \! \otimes \! \bm{b} \cdot \bm{\sigma} + \sum_{i,j=1}^3 T_{ij} \sigma_i \! \otimes \! \sigma_j \Biggr),
\end{equation}
where $\bm{\sigma} = [\sigma_1, \sigma_2, \sigma_3]^\intercal$; $\bm{a}, \bm{b} \in \R^3$; and $\bm{T} \in \R^{3\cross3}$. Denoting the diagonal elements of $\bm{T}$ by the vector $\bm{t}$, with components $t_i$, a Bell diagonal state $\rho$ can be represented by the simpler form~\cite{Garding2021}:
\begin{equation}
    \label{eq:pauli_rep}
    \rho = \frac{1}{4} \Biggl(\Id \otimes \Id + \sum_{i=1}^3 t_i \sigma_i \otimes \sigma_i\Biggr).
\end{equation}
Interpreting $\bm{t} \in [-1, 1]^3$ as a point in $\R^3$, the constraint that $\rho$ is physical corresponds to the condition that $\bm{t}$ lies inside the tetrahedron whose vertices correspond to the Bell states, as illustrated in Fig.~\ref{fig:figure1}. Furthermore, the set of separable Bell diagonal states, which satisfy the condition $\abs{t_1}+\abs{t_2}+\abs{t_3} \leq 1$, forms an octahedron within this geometric representation~\cite{Lang2010}.

The relation between the parameters $\bm{\theta}$ in the convex sum representation of Eq.~\eqref{eq:spectral_decomposition} and $\bm{t}$ in the Pauli representation of Eq.~\eqref{eq:pauli_rep} is given by~\cite{Garding2021}:
\begin{equation}
    \setlength\arraycolsep{2pt}
    \label{eq:rel}
    \bm{\theta} = \frac{1}{4} \begin{bmatrix*}[r]
    1 & \text{-}1 & 1 \\
    \text{-}1 & 1 & 1 \\
    1 & 1 & \text{-}1 \\
    \text{-}1 & \text{-}1 & \text{-}1 
    \end{bmatrix*} \bm{t} + \frac{1}{4} \begin{bmatrix*}[r]
    1\\
    1\\
    1\\
    1 
    \end{bmatrix*}, \quad \
    \bm{t} = \begin{bmatrix*}[r]
    1 & \text{-}1 & 1 & \text{-}1 \\
    \text{-}1 & 1 & 1 & \text{-}1 \\
    1 & 1 & \text{-}1 & \text{-}1
    \end{bmatrix*} \bm{\theta}. 
\end{equation}

Representing the output of a noisy Bell pair generation process as a Bell diagonal state is valid when the noise can be modeled as a Pauli channel $\mathcal{P}$,
\begin{equation}
\label{eq:PauliChannels}
    \mathcal{P}(\rho) = \sum_{i=0}^3\sum_{j=0}^3 p_{ij} \left(\sigma_i \otimes \sigma_j \right) \rho \left( \sigma_i \otimes \sigma_j \right),
\end{equation}
where $\sigma_0$ is the identity $\Id$ and the coefficients $p_{ij}$ form a probability distribution. Crucially, Pauli channels preserve the Bell diagonal structure of a state. For example, applying dephasing noise to a maximally entangled Bell state results in a Bell diagonal state. This follows directly from applying Eq.~\eqref{eq:PauliChannels} to Eq.~\eqref{eq:pauli_rep} and using the property $\sigma_i \sigma_j \sigma_i = (-1)^{\langle i, j \rangle} \sigma_j$, where $\langle i, j \rangle$ equals 0 if $\sigma_i$ and $\sigma_j$ commute and 1 otherwise. Among Pauli channels, depolarization is particularly relevant for entanglement distribution protocols~\cite{Ecker2019, Avis2023, Mor2024, deBone2024}, and can therefore be analyzed within the framework developed in this paper. However, some entanglement generation protocols fall outside this model, e.g., single-photon interference~\cite{cabrillo1999creation}, and certain important noise models, such as amplitude damping, are not Pauli channels.

For cases where the output of a noisy Bell state generation is not a Bell diagonal state, there exists a scheme that combines four copies of an arbitrary two-qubit state $\rho$ such that, on average, they are represented by a Bell diagonal state~\cite{Bennett1996}. As described in Ref.~\cite{Bennett1996}, this is achieved by applying one of the local operations $\Id \otimes \Id$, $\sigma_x \otimes \sigma_x$, $\sigma_y \otimes \sigma_y$, or $\sigma_z \otimes \sigma_z$ to each copy of $\rho$. Averaging over those transformed density matrices,
\begin{equation}
 \Tilde{\rho} = \frac{1}{4} \sum_{i=0}^3 \left(\sigma_i \otimes \sigma_i \right) \rho \left(\sigma_i \otimes \sigma_i \right),
 \end{equation}
results in a density matrix $\Tilde{\rho}$ that is diagonal in the Bell basis. The scheme is analogous to Pauli twirling, where an arbitrary noise channel is projected onto a Pauli channel by randomly applying Pauli gates to different circuit instances and averaging the measurement outcomes over many circuits~\cite{Bennett1996, Geller2013}.

\subsection{State estimation}
\label{sec:estimation}
Given a quantum device that prepares an unknown state $\rho_0$, quantum state estimation is the task of finding a description $\hat{\rho}$ of $\rho_0$ that allows predicting the outcome probabilities of any measurement performed on the device~\cite{paris2004}. We discuss three types of estimators: direct inversion, maximum-likelihood estimation, and Bayesian mean estimation.

We note that it is possible to characterize states \cite{vsupic2020self} and gates \cite{Noller2025} without trust in the measurement devices. In the following, we do not consider this adversarial setup and assume trusted devices and parties. These are natural assumptions in a practical scenario where local devices can be characterized before engaging in network tasks, and parties can be trusted after following an authentication protocol.

In quantum state tomography, or direct inversion, the observed frequencies of measurement outcomes $x$ are interpreted as the outcome probability of the respective measurement. The direct inversion estimate, $\hat{\rho}_T$, is then obtained by inverting Born's rule~\cite{NielsenChuang2011}. This approach uniquely determines $\hat{\rho}_T$ but may yield an unphysical state with negative eigenvalues~\cite{blume2010optimal}.

Maximum-likelihood estimation~\cite{Hradil1997, James2001} consists of finding the state that maximizes the likelihood function $\mathcal{L}(\rho) = p(x|\rho)$ over the set of density operators $\mathcal{S}(\Hil)$:
\begin{equation}
    \hat{\rho}_M = \arg \max_{\rho \in \mathcal{S}(\Hil)} \mathcal{L}(\rho).
\end{equation}
This method ensures that the estimated state is physical, but may result in rank-deficient estimates~\cite{blume2010optimal}, meaning some eigenvalues of the estimated state $\hat{\rho}$ are zero. This poses a problem, as it assigns zero probability to certain possible measurements, an outcome that is unjustifiable given a finite number of measurements and incompatible with error bars~\cite{blume2010optimal}.

The problem of rank-deficient estimates is avoided in Bayesian mean estimation~\cite{blume2010optimal, Granade2016, Lukens2020}, where a distribution over possible states is computed to reflect their relative plausibility. Given a prior distribution over the states $\pi(\rho)$, the posterior distribution $p(\rho|x)$ is obtained by applying Bayes' theorem:
\begin{equation}
\label{eq:BayesTheorem}
    p(\rho|x) = \frac{p(x|\rho) \pi(\rho)}{\int p(x|\rho) \pi(\rho) d\rho}.
\end{equation}
The Bayesian mean estimate denoted by $\hat{\rho}_B$ is then the mean state with respect to this posterior distribution:
\begin{equation}
    \hat{\rho}_B (x) = \int \rho \, p(\rho|x) \, d\rho.
\end{equation}
Since this estimation is based on a distribution, standard statistical tools, such as the covariance matrix, can be used to compute the uncertainty of the estimate~\cite{blume2010optimal}.

\subsection{Estimation risk}
\label{sec:risk}
The closeness of an estimate $\hat{\rho}(x)$ to the true state $\rho_0$ can be quantified by a loss function $L\left[\rho_0, \hat{\rho}(x)\right]$, which satisfies $L \left(\rho_0, \rho_0 \right) = 0$ and $L\left[\rho_0, \hat{\rho}(x)\right] \geq 0$, with equality if and only if $\rho_0 = \hat{\rho}(x)$~\cite{Quadeer2019}. In our analysis of Bell diagonal state estimation, where $\rho_0 = \sum_i \theta_i \Psi_i$ and $\hat{\rho}(x) = \sum_i \hat{\theta}_i(x) \Psi_i$, we define the loss function in terms of the Hilbert-Schmidt (HS) distance, which in this case is proportional to the mean-square error of the estimated vector $\hat{\bm{\theta}}$:
\begin{equation}
\label{eq:HS-distance_t}
    L\left[\rho_0, \hat{\rho}(x)\right]  = \Tr{\abs{\hat{\rho}(x) - \rho_0}^2} = \sum_{i=1}^4 \left[\hat{\theta}_i(x) - \theta_i\right]^2 \!.
\end{equation}

For a fixed $\rho_0$, the risk $ R\left(\rho_0, \hat{\rho}\right)$ of an estimator $\hat{\rho}$ describes the average closeness of the estimate over the space of possible measurement outcomes $X$~\cite{lehmann1998}:
\begin{equation}
    R\left(\rho_0, \hat{\rho}\right) = \Exc{X|\rho_0}{L \left[\rho_0, \hat{\rho}(X)\right]}.
\end{equation}
In terms of the Hilbert-Schmidt distance, the risk takes the form
\begin{equation}
\label{eq:HS_risk}
     R\left(\rho_0, \hat{\rho}\right) = \sum_{i=1}^4 \Exc{X|\bm{\theta}}{\left[\hat{\theta}_i(X) - \theta_i \right]^2}.
\end{equation}
For unbiased estimators, meaning $\mathbb{E}(\hat{\bm{\theta}}) = \bm{\theta}$ for all $\bm{\theta}$~\cite{demkowicz2020multi}, the sum of mean-square errors in Eq.~\eqref{eq:HS_risk} corresponds to the sum of the variances of the estimated parameters
\begin{equation}
\label{eq:risk_variance}
    R\left(\rho_0, \hat{\rho}\right) = \sum_{i=1}^4 \Var{\hat{\theta}_i} = \frac{1}{4} \sum_{i=1}^3 \Var{\hat{t}_i}.
\end{equation}

One way to eliminate the dependence of the estimator's risk on the generally unknown true state $\rho_0$ is to average the risk over a distribution $\pi(\rho)$ of possible states. This yields the average risk $r(\pi, \hat{\rho})$, defined as:
\begin{equation}
    \label{eq:avg_risk}
    r(\pi, \hat{\rho}) = \int R\left(\rho, \hat{\rho}\right) \pi(\rho) d\rho.
\end{equation}

Assuming a known prior distribution $\pi(\rho)$, the Bayesian mean estimator minimizes the average risk when the distance is defined by a Bregman divergence~\cite{Banerjee2005}. Hence, for the class of Bregman divergences, which includes the Hilbert-Schmidt distance, the average risk of the Bayesian mean estimator provides a lower bound for the average risk of any estimator.

The fidelity $F(\hat{\rho}, \rho_0)$ between the estimate $\hat{\rho}$ and the true state $\rho_0$ can be bounded in terms of the Hilbert-Schmidt distance $L(\hat{\rho}, \rho_0)$ by combining norm inequalities for the trace distance~\cite{Fuchs1999, Coles2019}:
\begin{equation}
    \left[1 - \sqrt{L(\hat{\rho}, \rho_0)} \right]^2 \leq F(\hat{\rho}, \rho_0) \leq 1 - \frac{1}{2} L(\hat{\rho}, \rho_0).
\end{equation}

\subsection{Quantum Cramér-Rao Bound}
\label{sec:QCR}
A lower bound for the variance of unbiased estimators can be found using the quantum Cramér-Rao bound~\cite{Helstrom1969, liu2020quantum}. This bound states that for a density matrix $\rho(\bm{y})$ in which a vector of unknown parameters $\bm{y}$ is encoded, the covariance matrix $\Cov{\hat{\bm{y}}}$ of an unbiased estimator $\hat{\bm{y}}$ satisfies:
\begin{equation}
\label{eq:QCR_bound}
    \Cov{\hat{\bm{y}}} \geq \frac{1}{N} \mathcal{F}_{\bm{y}}^{-1}, 
\end{equation}
where $\mathcal{F}_{\bm{y}}$ is the quantum Fisher information matrix and $N$ the number of repetitions of the experiment. For the total variance, this implies $\sum_i \Var{\hat{y}_i} \geq \frac{1}{N} \tr{\mathcal{F}_{\bm{y}}^{-1}}$~\cite{liu2020quantum}. The same bound holds when, instead of performing $N$ experiments on $\rho(\bm{y})$, one conducts a collective measurement on $\rho(\bm{y})^{\otimes N}$. This is because $\mathcal{F}_{N, \bm{y}} = N \mathcal{F}_{1, \bm{y}}$, where $\mathcal{F}_{N, \bm{y}}$ is the Fisher information of the state $\rho(\bm{y})^{\otimes N}$ and $\mathcal{F}_{1, \bm{y}}$ of the state $\rho(\bm{y})$~\cite{Hayashi2008}. Finally, alternative scalar bounds to the total variance can be obtained by introducing a positive-semidefinite matrix $W$, leading to the inequality $\tr{W\Cov{\hat{\bm{y}}}} \geq \tr{W \mathcal{F}_{\bm{y}}^{-1}}/N$~\cite{Albarelli2020, Notarnicola2024}.


\section{Results}

In this section we compare and discuss the risk of different Bell diagonal state estimators. Before considering specific types of measurements, we derive a general lower bound on the risk of unbiased estimators by applying the quantum Cramér-Rao bound.

\begin{lemma}
    Given $N$ copies of a Bell diagonal state $\rho_0 = \Sigma_i \theta_i \Psi_i$, the risk in terms of the Hilbert-Schmidt distance of an unbiased estimator $\hat{\rho}$ is lower bounded by
    \begin{equation}
    \label{eq:QFI}
        R\left(\rho_0, \hat{\rho}\right) \geq \frac{1 - \sum_{i=1}^4 \theta_i^2}{N}.
    \end{equation}    
\end{lemma}
\begin{proof}
    The idea of the proof is to calculate the quantum Fisher information for a Bell diagonal state $\rho_0 = \Sigma_i \theta_i \Psi_i$ and apply the quantum Cramér-Rao bound to Eq.~\eqref{eq:risk_variance}. $\rho_0$ can be parametrized with three independent parameters $\bm{t}$ according to Eq.~\eqref{eq:pauli_rep}. In the spectral decomposition $\rho_0(\bm{t}) = \sum_i \theta_i(\bm{t}) \Psi_i$ only the eigenvalues $\theta_i$ but not the eigenstates $\Psi_i$ depend on $\bm{t}$. This leads to a simple expression for the elements of the quantum Fisher information matrix $\mathcal{F}_{\bm{\theta}}$~\cite[Theorem 2.2]{liu2020quantum}:
    \begin{equation}
        \left[\mathcal{F}_{\bm{\theta}}\right]_{ab} = \sum_{\theta_i \in \mathcal S} \frac{(\frac{\partial}{\partial t_a} \theta_i) (\frac{\partial}{\partial t_b} \theta_i)}{\theta_i},
    \end{equation}
    where $\mathcal{S}=\{\theta_i \in \{\theta_i\} | \theta_i \neq 0\}$ is the support of $\rho_0$. Evaluating $\mathcal{F}_{\bm{\theta}}$ with Eq.~\eqref{eq:rel} and using the quantum Cramér-Rao bound of Eq.~\eqref{eq:QCR_bound} yields
    \begin{equation}
        \sum_{i=1}^3 \Var{\hat{t}_i} \geq 4 \, \frac{1 - \sum_{i=1}^4 \theta_i^2}{N}.
    \end{equation}
    Finally, applying this result to Eq.~\eqref{eq:risk_variance} provides a lower bound for the Hilbert-Schmidt distance risk of any unbiased estimator $\hat{\rho}$ and any set of measurements on the $N$-copy Bell diagonal state $\rho_0(\bm{t})^{\otimes N}$.
\end{proof}    

\subsection{Bell state measurements}
\label{sec:BellStateMeas}

In the absence of any practical limitations, Bell state measurements are the optimal choice for estimating Bell diagonal states. Since the Bell basis is the eigenbasis of the states in the ensemble, the estimation reduces to a completely classical statistical model~\cite{Suzuki2019, Albarelli2020}. However, Bell state measurements are not practically realizable in many distributed scenarios. We include them here to illustrate how the bound in Eq.~\eqref{eq:QFI} can be achieved and to provide intuition as to why this is not possible with separable strategies.

Formally, the measurement process corresponds to measuring the Bell diagonal state $\rho_0(\bm{\theta})$ using the positive operator-valued measure (POVM) $\M = \{\Psi_i\}_{i=1}^4$. The probability of observing outcome $\Psi_i$ is given by $\Tr{\Psi_i \rho(\bm{\theta})} = \theta_i$. Furthermore, the number of measurement outcomes $\bm{X} = (X_1, X_2, X_3, X_4)$ obtained from $N$ Bell state measurements on the Bell diagonal state $\rho_0\left(\bm{\theta}\right)$ follows a multinomial distribution with the parameters $N$ and $\bm{p} = \bm{\theta}$. 

\begin{figure}
  \centering
  \includegraphics[width=\columnwidth]{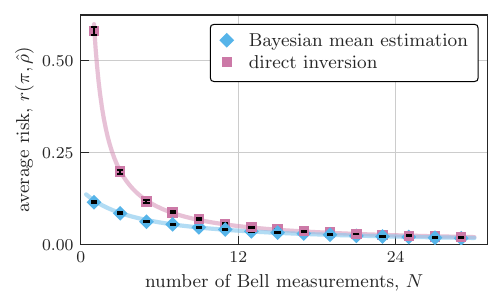}
  \caption{Average risk in terms of the Hilbert-Schmidt distance over a uniform prior for Bell state measurements plotted against the number of measurements $N$ for Bayesian mean estimation and direct inversion. The solid lines correspond to the analytical results of Eqs.~\eqref{eq:avg_risk_bsm_di} and \eqref{eq:avg_risk_bsm_b}. Each data point corresponds to the average over $1000$ samples. For each sample, the true state $\rho_0$ is drawn uniformly and $N$ Bell state measurements are simulated. To approximate the posterior distribution in the Bayesian mean estimation, the state space is discretized with $10^4$ states.}
  \label{fig:figure2}
\end{figure}

The expectation value of $X_i$ is given by $\Ex{}{X_i} = N \theta_i$. Inverting this relationship leads to the direct inversion estimate:
\begin{equation}
    \label{eq:BSM-di}
    \hat{\rho}_T(\bm{x}) = \sum_{i=1}^{4} \frac{x_i}{N} \Psi_i.
\end{equation}
Note that since $\sum x_i = N$ and $x_i > 0$, the state $\hat{\rho}_T = \sum \hat{\theta}_i \Psi_i$ is a physical state for all possible measurement outcomes $\bm{x}$. Hence, direct inversion estimation and maximum-likelihood estimation coincide in this case, as visible in Fig.~\ref{fig:figure4}~\cite{blume2010optimal}. 

\begin{lemma}[Direct inversion, Bell state meas.]
Given the outcome of $N$ Bell state measurements of the Bell diagonal state $\rho_0 = \Sigma \theta_i \Psi_i$, the risk in terms of the Hilbert-Schmidt distance for the direct inversion estimator $\hat{\rho}_T$ is,
\begin{equation}
    \label{eq:risk_bsm_di}
    R\left[\rho_0(\bm{\theta}), \hat{\rho}_T\right] = \frac{1 - \sum_{i=1}^4 \theta_i^2}{N},
\end{equation}
and the average risk for a uniformly distributed $\bm{\Theta}$ is
\begin{equation}
\label{eq:avg_risk_bsm_di}
    r\left(\pi, \hat{\rho}_T\right) = \frac{3}{5N}.
\end{equation}
\end{lemma}
\begin{proof}
    As evident from Eq.~\eqref{eq:BSM-di}, $\hat{\rho}_T$ is unbiased. Therefore, Eq.~\eqref{eq:risk_variance} can be used to calculate the risk. From Eq.~\eqref{eq:BSM-di} we observe $N^2 \, \Var{\hat{\theta}_i} = \Var{X_i}$. Moreover, $\Var{X_i} = N \theta_i (1- \theta_i)$ because each $X_i \in \bm{X}$ is binomially distributed with the parameters $N$ and $\theta_i$. Finally, substituting $\Var{\hat{\theta}_i} = (\theta_i - \theta_i^2)/N$ into Eq.~\eqref{eq:risk_variance} and using that $\Sigma_i \theta_i = 1$ leads to the expression for the risk. The average risk is then calculated according to
    \begin{equation}
    \label{eq:avg_risk_calc}
        r\left(\pi, \hat{\rho}\right) = \int_{\bm{\Theta}} R\left[\rho_0(\bm{\theta}), \hat{\rho}\right] d\pi(\bm{\theta}).
    \end{equation}
\end{proof}

\begin{corollary}
    Bell state measurements are optimal for unbiased Bell diagonal state estimation in the sense that Eq.~\eqref{eq:risk_bsm_di} achieves the lower bound obtained in Eq.~\eqref{eq:QFI}. There is no advantage in doing a collective measurement on $\rho_0(\bm{\theta})^{\otimes N}$ as the quantum Cramér-Rao bound is met by $N$ separable Bell state measurements.
\end{corollary}

Compared to direct inversion or maximum-likelihood estimation, Bayesian mean estimation introduces an additional dependence on a prior distribution $\pi(\theta)$. While the results in this section are derived using the uniform prior $\Theta \sim \Dir{\bm{1}_4}$, the same analysis can be done for other Dirichlet distributions $\Dir{\bm{\alpha}}$, such as the Jeffreys prior.

\begin{lemma}[BME, Bell state meas.]
    Given the outcome of $N$ Bell state measurements of the Bell diagonal state $\rho_0 = \Sigma \theta_i \Psi_i$ and assuming a uniform prior $\pi(\theta)$ on the simplex $\bm{\Theta}$, the risk in terms of the Hilbert-Schmidt distance for the Bayesian mean estimator $\hat{\rho}_B$ is
    \begin{equation}
    \label{eq:risk_bsm_b}
        R\!\left[\rho_0(\bm{\theta}), \hat{\rho}_B\right] \!=\! \frac{N\!\left(\!1\!-\!\sum_{i=1}^4 \theta_i^2 \right)\!+\!4\!\left(4 \sum_{i=1}^4 \theta_i^2\!-\!1\!\right)}{\left(4\!+\!N \right)^2}
    \end{equation}
    and the average risk for a uniformly distributed $\bm{\Theta}$ is,
    \begin{equation}
    \label{eq:avg_risk_bsm_b}
        r( \pi, \hat{\rho}_B) = \frac{3}{5( N + 4)}.
    \end{equation}
\end{lemma}

\begin{proof}
The uniform prior $\pi(\theta)$ on the simplex $\bm{\Theta}$ corresponds to the symmetric Dirichlet distribution $\Dir{\bm{1}_4}$~\cite{bernard2005}. Dirichlet distributions are conjugate priors for the multinomial distribution, meaning that if $(\!\bm{X}|\bm{\theta})$ follows a multinomial distribution and $\bm{\Theta} \sim \Dir{\bm{\alpha}}$, then the posterior $(\bm{\Theta}|\bm{x})$ follows a $\Dir{\bm{\alpha} + \bm{x}}$ distribution~\cite{frigyik2010introduction}. Since the number of measurement outcomes $\bm{X}$ from $N$ Bell state measurements is distributed multinomially and $\pi(\theta)$ is the symmetric Dirichlet prior, we find that the posterior distribution $p(\bm{\theta}|\bm{x})$ corresponds to the probability density function of the distribution $\Dir{\bm{x} + \bm{1}_4}$, which is given by
\begin{equation}
    p(\bm{\theta}|\bm{x}) = \frac{1}{B(\bm{x} + \bm{1}_4)} \prod_{i=1}^4 \theta_i^{x_i},
\end{equation}
where $B(\bm{x} + \bm{1}_4) = \left[\prod_i \Gamma(x_i+1)\right] /\Gamma(N+4)$~\cite{frigyik2010introduction}.
From the posterior distribution, we can calculate the Bayesian mean estimate $\hat{\bm{\theta}}$ according to
\begin{equation}
    \hat{\theta}_i(\bm{x}) = \int_{\bm{\Theta}} \theta_i \, p(\bm{\theta}|\bm{x}) \, d\bm{\theta} = \frac{x_i + 1}{N+4}.
\end{equation}
To evaluate Eq.~\eqref{eq:HS_risk} we use that since $X_i$ is distributed binomially with the parameters $N$ and $p=\theta_i$, $\Ex{}{X_i} = N \theta_i$ and $\Ex{}{X_i^2} = N \theta_i + (N^2 - N) \theta_i^2$. Simplifying and using $\Sigma_i \theta_i = 1$ leads to the expression for the risk. The average risk is then calculated according to Eq.~\eqref{eq:avg_risk_calc}.
\end{proof}

In the regime $N < 16$, the risk is minimal for the completely mixed state and maximal for the Bell states. This occurs because the mean over the prior $\int \rho(\bm{\theta}) d\pi(\bm{\theta}) = \frac{1}{4} \Id$ corresponds to the completely mixed state. For smaller $N$, the choice of the prior becomes more influential on the Bayesian mean estimate. Notably, unlike the direct inversion estimate $\hat{\rho}_T$, the Bayesian mean estimate $\hat{\rho}_B$ is always of full rank for finite $N$, since $\hat{\theta}_i > 0$ for all $\bm{x} \in \bm{X}$. For large $N$, the asymptotic behavior of the risk aligns with that of direct inversion in Eq.~\eqref{eq:avg_risk_bsm_di}, as visible in Fig.~\ref{fig:figure2}.

\subsection{Separable measurements}
Implementing a Bell state measurement on a bipartite state shared between parties A and B involves executing a nonlocal two-qubit operation. When these parties are spatially separated, e.g., two nodes of a quantum network, performing nonlocal operations necessitates not only classical communication and local gates but also preshared entanglement~\cite{Eisert2000}. If the resources to perform nonlocal two-qubit gates are not available, the alternative for estimating the shared state is to use measurements based on local operations and classical communication (LOCC).

To find an optimal set of separable measurements, we begin by discussing their measurement statistics on maximally entangled states. Let $P, Q  \in \text{Pos}(\mathbb{C}^{2})$ be rank-$1$ projection operators and consider $\rho_{AB}$, a maximally entangled pure state in the space $\mathbb{C}^{2} \otimes \mathbb{C}^2$. Consequently, $\Trr{A}{\rho_{AB}} = \Trr{B}{\rho_{AB}} = \Id/2$~\cite{Goyeneche2015}. In this case, we obtain the following upper bound on the expectation value of the separable operator $P_A \otimes Q_B$
\begin{equation}
\begin{aligned}
    \Tr{(P_A \otimes Q_B) \rho_{AB}} &\leq \Tr{(P_A \otimes \Id_B) \rho_{AB}} \\
    &= \Trr{A}{P_A \frac{\Id}{2}} = \frac{1}{2}.
\end{aligned}
\end{equation}
This bound directly relates to the fact that the marginals of $\rho_{AB}$ are completely mixed, meaning the marginal outcome probabilities contain no information about the maximally entangled state.

It is not possible to improve this upper bound by conducting collective separable rank-$1$ measurements $R, S \in \text{Pos}(\mathbb{C}^{2^n})$ on the state $\rho_{AB}^{\otimes n}$. Following similar steps to those above, we obtain $\tr{(R_A \otimes S_B) \rho_{AB}^{\otimes n}} \leq \frac{1}{2^n}$, which corresponds to the bound for $R = P^{\otimes n}$ and $S = Q^{\otimes n}$.

These considerations also apply to Bell diagonal states, as they are a convex sum of maximally entangled states.

\subsubsection{Pauli parity check}
\label{sec:ParityChecks}
We now turn explicitly to measurements exploiting the specific structure of Bell states $\rho(\bm{t})$ in Eq.~\eqref{eq:pauli_rep}. We express the rank-$1$ projection operators $P$ and $Q$ in the Bloch representation as $P = \frac{1}{2}(\Id + \bm{p} \cdot \bm{\sigma})$ and $Q = \frac{1}{2}(\Id + \bm{q} \cdot \bm{\sigma})$, where $\bm{p}, \bm{q}\in\mathbb{R}^3$, $\abs{\bm{p}} = \abs{\bm{q}} = 1$, and $\bm{\sigma} = [\sigma_1, \sigma_2, \sigma_3]^\intercal$. Using the Pauli representation of Eq.~\eqref{eq:pauli_rep}, we find
\begin{equation} \label{eq:trace-sep-general}
    \Tr{(P \otimes Q) \rho(\bm{t})} = \frac{1}{4}\left(1 + \sum_{i=1}^3 p_i q_i t_i\right).
\end{equation} 
The only $\mathbf{p}$ and $\mathbf{q}$ for which the outcome probabilities take on the extreme values $0$ or $1/2$ for all four Bell states are those with $\mathbf{p} = \pm \mathbf{q}$ and $\mathbf{p}$ (or $-\mathbf{p}$) being a standard basis vector in $\mathbb{R}^3$. This corresponds to measuring both qubits in the same Pauli basis. When measuring both qubits in the eigenbasis of $\sigma_i$, the probability of finding both qubits in the up eigenstate is the same as finding both qubits in the down eigenstate and, similarly, the probabilities of up-down and down-up eigenstates are identical. Consequently, we can interpret this measurement as a parity check in the basis $i$ with POVM elements $1/2 (\Id + \sigma_i \otimes \sigma_i)$ and $1/2 (\Id - \sigma_i \otimes \sigma_i)$. We refer to these measurements as Pauli parity checks.

Furthermore, for a general state $\rho$, represented as in Eq.~\eqref{eq:gen_state_pauli_rep} by $\bm{a}$, $\bm{b}$, and $\bm{T}$, the outcome probabilities 
\begin{equation}
    \begin{aligned}
        \Tr{1/2 (\Id + \sigma_i \otimes \sigma_i) \rho} &= 1/2 (1 + T_{ii})\\
        \Tr{1/2 (\Id - \sigma_i \otimes \sigma_i) \rho} &= 1/2 (1 - T_{ii})
    \end{aligned}
\end{equation} 
only depend on the diagonal elements of $\bm{T}$. Therefore, Pauli parity checks allow the estimation of the diagonal entries of any two-qubit density matrix represented in the Bell basis.

\subsubsection{Distinguishability of Bell diagonal states}

The precision of parameter estimation is fundamentally tied to how distinguishable quantum states are. The more sensitive a state is to changes in a parameter, the easier it is to differentiate between nearby states. This effect is quantifiable in terms of the quantum Fisher information \cite{BraunsteinCaves_PRL94}. Here we consider pairs of Bell diagonal states and ask if the optimal measurement to distinguish them is locally implementable. 

An upper bound for the success probability of distinguishing a pair of states $\rho$ and $\phi$ is given by the Helstrom bound~\cite{Helstrom1969}:
\begin{equation}
    p \leq \frac{1}{2} \left(1 + \frac{1}{2} \tr{\abs{\rho - \phi}}\right).
\end{equation}
The measurement that achieves this bound is a projective measurement onto the positive and negative part of the Helstrom matrix $\rho - \phi$~\cite{Puchala2016}.

For $\rho = \Sigma_i \rho_i \Psi_i$ and $\phi = \Sigma_i \phi_i \Psi_i$ being Bell diagonal states, the matrix $\rho - \phi$ is diagonal in the Bell basis. Consequently, the optimal POVM corresponds to
\begin{equation}
    M = \left\{ \sum_{i: \rho_i - \phi_i > 0} \Psi_i, \sum_{i: \rho_i - \phi_i \leq 0} \Psi_i\right\}.
\end{equation}
This POVM is straightforwardly implementable with Bell state measurements. More interesting is the case where we only consider LOCC implementable measurements. The even mixture of two Bell states $1/2(\Psi_i + \Psi_j)$, with $j \neq i$, is separable, as illustrated in Fig.~\ref{fig:figure1}. Therefore, POVMs of the form $M = \{\Psi_i + \Psi_j, \Psi_k + \Psi_l\}$, where $i$, $j$, $k$, and $l$ are required to be different to ensure the elements sum up to the identity, are LOCC implementable. The respective measurements to implement these POVMs correspond to the parity check measurements described in the preceding section. On the other hand, POVMs of the form $M = \{\Psi_i + \Psi_j + \Psi_k, \Psi_l\}$ are not LOCC implementable as they discriminate one Bell state from the other three Bell states, which is known to be not possible for LOCC implementable measurements~\cite{Ghosh2001}. Consequently, a LOCC implementable POVM that optimally distinguishes the Bell diagonal states $\rho$ and $\phi$ exists if $\abs{\{ \, i \, | \, \rho_i - \phi_i > 0 \}} \leq 2$ and $\abs{\{ \, i \, | \, \rho_i - \phi_i < 0 \}} \leq 2$. This is always the case for $\text{rank}(\rho - \phi) \leq 3$. Therefore, the convex sets spanned by three Bell diagonal states, in the geometric representation of Fig.~\ref{fig:figure1}, represented by the faces of the tetrahedron, are sets of states where each pair within the set can be optimally distinguished by separable measurements. For $\text{rank}(\rho - \phi) = 4$, this is not necessarily the case. An example of two Bell diagonal states that are not optimally distinguishable with LOCC is the states $\rho = \Id / 4$ and $\phi = 2/5 \Psi_1 + 1/5 (\Psi_2 + \Psi_3 + \Psi_4)$. This example also demonstrates that the separability of both states, $\Sigma \abs{t_i} \leq 1$, does not imply optimal distinguishability with separable measurements. 

Furthermore, for a pair of locally optimally distinguishable states $\{ \rho, \sigma \}$, any convex combinations of these states, given by $\psi = p \rho + (1 - p) \sigma$ and $\phi = q \rho + (1 - q) \sigma$, with $p, q \in [0, 1]$, remain optimally distinguishable using local measurements. This follows from the fact that their Helstrom matrix is simply a scaled version of the original one $\psi - \phi = (p - q)(\rho - \sigma)$.

Let us say we use the locally implementable Helstrom measurement to estimate such states. While it is clear that this measurement is not optimal for estimating these states, one could quantify how good it is by defining a ratio between the classical Fisher information for the Helstrom measurement and the quantum Fisher information. Quantum Fisher information is the maximum Fisher information achievable and thus this ratio would capture how far off we are from the optimal measurement.

\subsubsection{Estimating Bell diagonal states}

\begin{figure}
  \centering
  \includegraphics[width=\columnwidth]{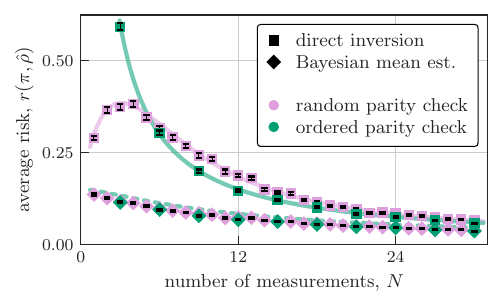}
  \caption{Average risk in terms of the Hilbert-Schmidt distance over a uniform prior for Pauli parity checks plotted against the number of measurements $N$ for Bayesian mean estimation and direct inversion. The solid lines correspond to the analytical results of Eqs.~\eqref{eq:avg_risk_dirinv} and \eqref{eq:avg_risk_dirinv_ordered} and green dashed line represents the upper bound of Eq.~\eqref{eq:bound_BME_sep}. Each data point corresponds to the average over $1000$ samples. For each sample, the true state $\rho_0$ is drawn uniformly and $N$ Pauli parity check measurements are simulated. To approximate the posterior distribution in the Bayesian mean estimation, the state space is discretized with $10^4$ states.}
  \label{fig:figure3}
\end{figure}

The POVM that corresponds to randomly selecting a basis $x$, $y$, or $z$ and performing a Pauli parity check in that basis can be expressed as
\begin{equation}
\label{eq:Pauli_BDS}
   \mathcal{M} = \left\{ \frac{1}{6} \left[\Id + (-1)^j \sigma_i \otimes \sigma_i \right] \right\}_{i\in\{1,2,3\}, j\in\{0,1\}},
\end{equation}
where the indices $i$ correspond to the basis of the parity check and $j$ to even and odd parity. We refer to this measurement as random Pauli parity checks. The number of outcomes $M_{i, j}$ obtained from $N$ measurements of the state $\rho(\bm{t})$ is denoted by $X_{i,j}$. Furthermore, we introduce the random variable $Y_i = X_{i, 0} + X_{i, 1}$ that represents the number of times the POVM outcome provided information about the parity in basis $i$; $Y_i$ is binomially distributed with the parameters $N$ and $p_i = 1/3$.

\begin{lemma}[Direct inversion, parity checks]
    Given the outcome of $N$ random Pauli parity check measurements of the Bell diagonal state $\rho_0 = \Sigma_i \theta_i \Psi_i = 1/4 (\Id \otimes \Id + \Sigma_i t_i \sigma_i \otimes \sigma_i)$, consider the direct inversion estimator
    \begin{equation}
    \label{eq:directInversion_PauliBDS}
      \hat{t}_i =\begin{cases}
        0, & \text{if $y_{i} = 0$}\\
        \frac{2 x_{i, 0}}{y_{i}} - 1, & \text{else},
      \end{cases}
    \end{equation}
    where $y_i$ denotes the number of measurements in the basis $i$ and $x_{i,0}$ the number of measurements in the basis $i$ with outcome $0$. The risk $R\left(\rho_0(\bm{\theta}), \hat{\rho}_T\right)$ of the estimator in terms of the Hilbert-Schmidt distance is
    \begin{equation}
    \label{eq:risk_inv_exact}
        R = \biggl( \frac{2}{3}\biggr)^{ N} \left\{ \biggl(1 - \sum_{i=1}^4 \theta_i^2 \biggr) \left[ \frac{N}{2} g(N) - 1 \right] + \frac{3}{4}\right\} ,
    \end{equation}
    where the function $g(N)$ is the generalized hypergeometric function ${}_{3}F_{2} (1, 1, 1-N; 2, 2; -1/2)$.
    The average risk for a uniformly distributed $\bm{\Theta}$ is
    \begin{equation}
    \label{eq:avg_risk_dirinv}
        r\left(\pi, \hat{\rho}_T\right) = \left(\frac{2}{3}\right)^{ N}\left[\frac{N}{5} g(N) + \frac{3}{20}\right].
    \end{equation}
\end{lemma}
\begin{proof}
To calculate the risk, we start by finding an expression for $\Ex{\bm{X}|\bm{Y}}{\hat{t}_i - t_i}^2$,
\begin{equation}
\label{eq:dirinv_random_pc}
  \Exb{\bm{X}|\bm{Y}}{\left( \hat{t}_i - t_i \right)^2} =\begin{cases}
    t_i^2, & \text{if $y_{i} = 0$}\\
    \frac{1 - t_i^2}{y_i} & \text{else},
  \end{cases}
\end{equation}
where the expression for the case $y_i \neq 0$ is obtained by the observation that for a fixed $y_i$, $X_i$ is binomially distributed with the parameters $y_i$ and $p = (1 + t_i)/2$. Given the binomial distribution of $Y_i$ we get
\begin{equation}
    \begin{aligned}
        &\Exc{\bm{Y}}{\Exb{\bm{X}|\bm{Y}}{\left( \hat{t}_i - t_i \right)^2}} = \left( \frac{2}{3} \right)^{ N} t_i^2 \\ & \quad + \sum_{y=1}^N \binom{N}{y} \left( \frac{1}{3} \right)^{ y} \left( \frac{2}{3} \right)^{ N-y} \frac{1 - t_i^2}{y}\\
        &= \left( \frac{2}{3} \right)^{ N} t_i^2 + \frac{1 - t_i^2}{2} N \left(\frac{2}{3}\right)^{ N} g(N),
    \end{aligned}
\end{equation}
where the function $g(N)$ is the generalized hypergeometric function ${}_{3}F_{2} (1, 1, 1-N; 2, 2; -1/2)$.

Finally, we use $\sum_i \ex{( \hat{t}_i - t_i )^2} = 4 \sum_i \ex{( \hat{\theta}_i - \theta_i )^2}$ and $\sum_{i=1}^3 t_i^2 = 4 \sum_{i=0}^3 \theta_i^2 - 1$ and insert the expression into Eq.~\eqref{eq:HS_risk} to reach the equation in question. The average risk is then calculated according to Eq.~\eqref{eq:avg_risk_calc}.
\end{proof}

In Fig.~\ref{fig:figure3}, an interesting feature of the average risk function for random parity checks becomes apparent. For the first few measurements, the average risk increases before it eventually starts to decrease. This behavior can be understood by analyzing Eq.~\eqref{eq:dirinv_random_pc}. Calculating the average risk over a uniform prior $\pi(\bm{t})$, we find from $\Ex{}{t_i^2} = 1/5$ that $\Ex{}{t_i^2} \leq \Exb{}{(1-t_i^2)/y_i}$ for $y_i \leq 4$. Therefore, the average risk for the estimate $t_i$ is smaller when there was no parity check in basis $i$ than when there was one, two, or three. The average risk for the estimate $\hat{t}_i$ is the largest for $y_i = 1$.

A simpler expression for the risk is reached with the additional assumption that $y_i = N/3$, which corresponds to the scenario where we measure equally many times in the $x$, $y$, and $z$ basis. We refer to this measurement strategy as ordered Pauli parity checks.

\begin{lemma}[Direct inversion, ordered parity checks]
    Given the outcome of $N/3$, $N \in 3\mathbb{N}$, Pauli parity check measurements of the Bell diagonal state $\rho_0 = \Sigma \theta_i \Psi_i$ in each of the three directions $x$, $y$, and $z$, the risk in terms of the Hilbert-Schmidt distance of the direct inversion estimator $\hat{\rho}_T$ is
    \begin{equation}
        \label{eq:risk_p_bds_di}
        R\left[\rho_0(\bm{\theta}), \hat{\rho}_T\right]= \frac{3 - 3 \sum_{i=1}^4 \theta_i^2}{N}.
    \end{equation}
    The average risk for a uniformly distributed $\bm{\Theta}$ is,
    \begin{equation}
    \label{eq:avg_risk_dirinv_ordered}
        r\left(\pi, \hat{\rho}_T\right) = \frac{9}{5N}.
    \end{equation}
\end{lemma}
\begin{proof}
In the case of $N/3$ parity checks in basis $i$, $X_{i, 0}$ is binomially distributed with the parameters $N/3$ and $p_i = (1+t_i)/2$. From the expectation value $\Ex{}{X_{i, 0}} = N(1+t_i)/6$ we find the direct inversion estimate
\begin{equation}
    \hat{t}_i = \frac{6 x_{i, 0} - N}{N}.
\end{equation}
Due to $X_{i, 0}$ being distributed binomially $\Var{X_{i, 0}} = N(1 - t_i^2)/12$ and consequently $\Var{\hat{t}_i} = 3 (1 - t_i^2) / N$. As $\hat{\bm{t}}$ is unbiased, the risk can be calculated according to Eq.~\eqref{eq:risk_variance}. Using $\sum_{i=1}^3 t_i^2 = 4 \sum_{i=0}^3 \theta_i^2 - 1$ we reach the expression for the risk. The average risk is then calculated according to Eq.~\eqref{eq:avg_risk_calc}.
\end{proof}

For increasing $N$, the risk for the parity check where each basis is measured equally many times approaches the risk of the randomly chosen parity checks, as visible for the average risk in Fig.~\ref{fig:figure3}. Comparing the result of Eq.~\eqref{eq:risk_p_bds_di} to the risk obtained for direct inversion with Bell state measurements of Eq.~\eqref{eq:risk_bsm_di}, we see that for a fixed $N$ the risk is three times bigger for the LOCC implementable ordered Pauli parity checks than for the Bell state measurements. 

Next, we derive an upper bound on the average risk for the Bayesian mean estimator based on the measurement outcomes of $N$ ordered Pauli parity checks. To calculate the posterior distribution under a uniform prior, we must integrate the likelihood function over the set of all physical states, as visible in the denominator of Bayes' theorem in Eq.~\eqref{eq:BayesTheorem}. For Bell diagonal states expressed in the Pauli representation of Eq.~\eqref{eq:pauli_rep}, this requires integrating $\bm{t}$ over the tetrahedron illustrated in Fig.~\ref{fig:figure1}, which is generally challenging to perform analytically. To simplify, we introduce an approximation for the true posterior distribution by extending the domain to all states with $\bm{t} \in [-1, 1]^3$, thus covering both physical and nonphysical states. For this extended set, the integration boundaries are $[-1, 1]$ for each variable $t_i$, making calculations more tractable. Geometrically, this corresponds to integrating over the cube that contains the tetrahedron in Fig.~\ref{fig:figure1}. This approximation gets better the more the extended posterior distribution is localized within the set of physical states, which occurs the more the true state $\rho_0$ is mixed and the more measurements are conducted. Given that the Bayesian mean estimator minimizes average risk (as discussed in Sec.~\ref{sec:risk}), we use the average risk calculated over the extended posterior distribution to provide an upper bound on the average risk of the true Bayesian mean estimator.

\begin{lemma}[BME, ordered parity checks]
    Given the outcome of $N/3$, $N\in 3\mathbb{N}$, Pauli parity check measurements of the Bell diagonal state $\rho_0 = \Sigma \theta_i \Psi_i$ in each of the three directions $x$, $y$, and $z$, and a uniform prior $\pi(\theta)$ on the simplex $\bm{\Theta}$, the average risk in terms of the Hilbert-Schmidt distance for the Bayesian mean estimator $\hat{\rho}_B$ is upper bounded by
    \begin{equation}
    \label{eq:bound_BME_sep}
        r\left(\pi, \hat{\rho}_B\right) \leq \frac{N+3}{5 \left(\frac{N}{3} + 2\right)^2}.
    \end{equation}
\end{lemma}
\begin{proof} As the Bayesian mean estimator minimizes the average risk if a Bregman divergence serves as a loss function~\cite{Banerjee2005}, the average risk of any estimator upper bounds the average risk of the Bayesian mean estimator. In the following, the average risk is calculated for an adjusted Bayesian mean estimator, where the posterior distribution is calculated not on the set of physical Bell diagonal states, but on the states $\rho(\bm{t}) = 1/4(\Id + \Sigma t_i \sigma_i \otimes \sigma_i)$ with $t_i \in [-1,1]$.

Denoting by $x_{i, j}$ the number of measurement outcomes $j$ of the Pauli parity checks in basis $i$, the likelihood function $p[\bm{x}| \rho(\bm{t})]$ is given by
\begin{equation}
\begin{aligned}
\label{eq:llh-BME-PPC}
    p\left[\bm{x}|\rho(\bm{t})\right] &= \prod_{i=1}^3 \binom{N/3}{x_{i, 0}} \prod_{j=1}^2 \left[\frac{1 + (-1)^{j} t_i}{2}\right]^{x_{i,j}}\\
    &= \prod_{i=1}^3 \binom{N/3}{x_{i, 0}} \left(\frac{1 + t_i}{2}\right)^{x_{i,0}} \left(\frac{1 - t_i}{2}\right)^{x_{i,1}}.
    \end{aligned}
\end{equation}

To calculate the posterior over the prior $\pi(\bm{t}) = 1/8, \forall \bm{t} \in [-1, 1]^3$, we first evaluate
\begin{equation}
\label{eq:norm-BME-PPC}
    \begin{aligned}
    &p(\bm{x}) = \iiint p(\bm{x}|\rho) \pi(\rho) d\bm{t} \\ 
    &= \frac{1}{8} \prod_{i=1}^3 \binom{N/3}{x_{i, 0}} \int_{-1}^{1} \left(\frac{1 + t_i}{2}\right)^{x_{i,0}} \left(\frac{1 - t_i}{2}\right)^{x_{i,1}} dt_i \\
    &= \frac{1}{8} \prod_{i=1}^3 2 \binom{N/3}{x_{i, 0}} \frac{\Gamma(x_{i,0}+1) \, \Gamma(N/3-x_{i,0}+1)}{\Gamma(N/3+2)}.
    \end{aligned}
\end{equation}

\begin{figure*}[ht]
  \centering
  \includegraphics[width=\textwidth]{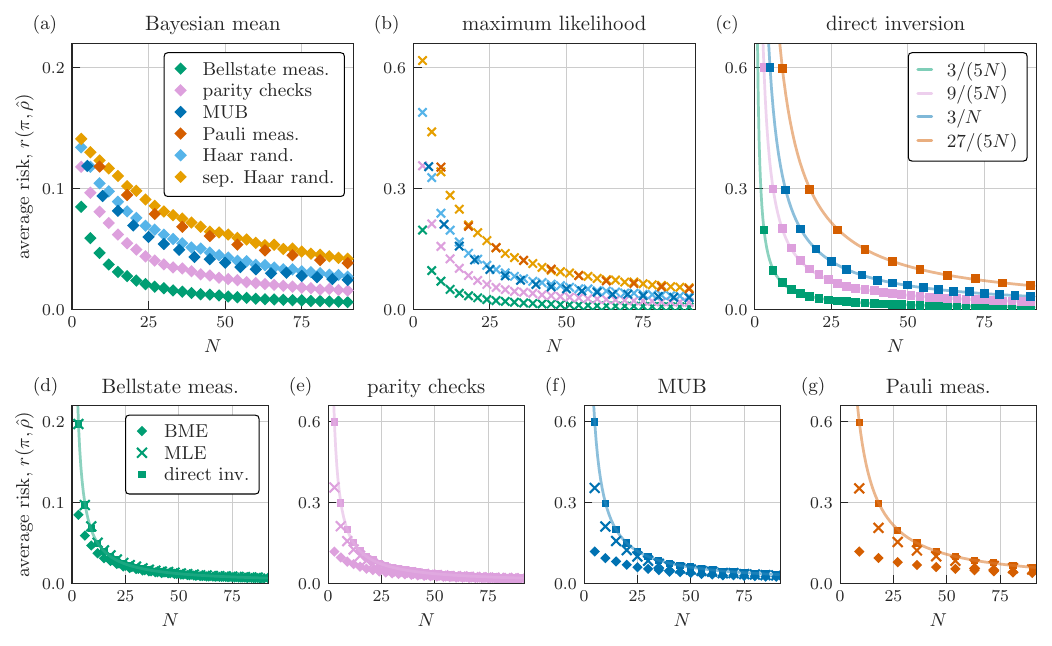}
  \caption{Average risk in terms of the Hilbert-Schmidt distance over a uniform prior plotted against the number of measurements $N$ for different types of estimators and different sets of measurements. The the risk of different measurements is compared for a fixed estimator: (a) Bayesian mean estimator, (b) maximum-likelihood estimator, and (c) direct inversion estimator. (d)-(g) Measurements are fixed to compare the estimators: (d) Bell state measurement, (e) parity checks, (f) MUB measurement, and (g) Pauli measurement. Each data point corresponds to an average of 4000 samples. The Pauli parity checks, Pauli measurements, and mutually unbiased basis measurements were conducted such that each of the three Pauli parity checks, five MUB measurements, and nine Pauli measurements was measured equally many times. To approximate the posterior distribution in the Bayesian mean estimation, the state space is discretized with $10^4$ states.}
  \label{fig:figure4}
\end{figure*}

Plugging Eqs.~\eqref{eq:llh-BME-PPC} and \eqref{eq:norm-BME-PPC} into Bayes' theorem of Eq.~\eqref{eq:BayesTheorem} leads to
\begin{equation}
    p\left[\rho(\bm{t})|\bm{x}\right] = \prod_{i=1}^3 \frac{\left(\frac{1+t_i}{2}\right)^{x_{i,0}}\left(\frac{1 - t_i}{2}\right)^{N/3 - x_{i,0}}}{\frac{\Gamma(x_{i,0}+1) \, \Gamma(N/3-x_{i,0}+1)}{\Gamma(N/3+2)}}.
\end{equation}

Replacing $(1 + t_i)/2$ with $z_i$ and $(1 - t_i)/2$ with $1 - z_i$, $p[\rho(\bm{z})|\bm{x}]$ is proportional to the probability density function of the product of three Beta distributions for three independent random variables $Z_i$ with the parameters $x_{i,0}+1$ and $x_{i,1}+1$. From the properties of the Beta distribution follows $\Ex{}{Z_i} = (x_{i,0}+1) / (N/3+2)$. As the estimate of the adjusted Bayesian mean estimator is the mean of the posterior, the estimates $\hat{t}_i$ correspond to
\begin{equation}
    \hat{t}_i = 2 \frac{x_{i, 0} + 1}{N/3 + 2} -1 = \frac{2 x_{i, 0} - N/3}{2 + N/3}.
\end{equation}
The risk is then calculated by evaluating
\begin{equation}
    R\left[\rho_0(\bm{t}), \hat{\rho}\right] = \frac{1}{4} \Exb{\bm{X}|\rho_0}{\sum_{i=1}^3 \left( \hat{t}_i(\bm{x}) - t_i\right)^2}.
\end{equation}
As $X_{i, 0}$ is distributed binomially with the parameters $N/3$ and $(1+t_i) / 2$, $\Ex{}{X_{i,0}} = N (1 + t_i) / 6$ and $\Ex{}{X_{i,0}^2} = N (1 - t_i^2)/12 + N^2 (1 + t_i)^2 / 36$. Furthermore, using $\Sigma t_i^2 = 4 \Sigma \theta_i^2 - 1$, we reach
\begin{equation}
    R\left[\rho_0(\bm{t}), \hat{\rho}\right] = \frac{N/3 -1 + (4 - N/3) \sum_{i=1}^4 \theta_i^2}{(N/2 +2)^2}.
\end{equation}

Taking the average of this risk over a uniform prior on the physical states leads to the average risk
\begin{equation}
    r = \frac{N+3}{5 (N/3 + 2)^2},
\end{equation}
which serves as an upper bound for the average risk of the true Bayesian mean estimator, where the prior and the posterior distribution are only over the set of physical states.
\end{proof}

As shown in Fig.~\ref{fig:figure3}, the calculated bound is relatively tight. For finite $N$, this upper bound is lower than the average risk for the direct inversion estimator. The bound approaches the average risk for the direct inversion estimator $r = 9 / (5N)$ for large $N$, as we would expect.

\subsection{Numerics}

In this section, we present the numerical results for the average risk of different sets of measurements and estimators. We compare them to the two sets of measurements, Bell state measurements and ordered parity checks, discussed analytically, and highlight some interesting observations. Beyond direct inversion and Bayesian mean estimation, we also include maximum-likelihood estimation (MLE), as it is one of the most popular choices to do state estimation (see Sec.~\ref{sec:estimation}). We consider six different sets of projective measurements. 
\begin{enumerate}[itemsep=0ex]
    \item Bell state measurements. These are measurements in the Bell basis as discussed in Sec.~\ref{sec:BellStateMeas}.
    \item Pauli parity checks. These are parity checks in the $x$, $y$, and $z$ basis as discussed in Sec.~\ref{sec:ParityChecks}. We focus on the ordered Pauli parity checks, where each basis is measured $N/3$ times.
    \item Mutually unbiased basis (MUB) measurements. For four dimensions, a set of five MUBs is informationally complete~\cite{wooters1989}. We use the set defined in Ref.~\cite{klappenecker2004}, where three bases are separable and two are maximally entangled. We measure in each of the five bases $N/5$ times.
    \item Pauli measurements. The measurement basis corresponds to the eigenbasis of the tensor product of two Pauli operators. We measure in each of the nine bases $N/9$ times.
    \item Haar random measurements. For each measurement, a unitary from the Haar measure on the unitary group of dimension $4$ is sampled. The measurement basis is then defined by the columns of the unitary.
    \item Separable Haar random measurements. Two unitaries from the Haar measure on the unitary group of dimension $2$ are sampled for each measurement. The measurement basis is then defined by the columns of the tensor product of the two unitaries. 
\end{enumerate}

In Figs.~\hyperref[fig:figure4]{\ref{fig:figure4}a}-\hyperref[fig:figure4]{\ref{fig:figure4}c} we compare different types of measurements for each of the three estimators. As expected, Bell state measurements yield the lowest average risk among the tested sets for all estimators. The second-best results are achieved by the ordered Pauli parity checks. Interestingly, measurements in the MUBs perform similarly to Haar-random measurements, while the results of Pauli measurements are similar to those of separable Haar-random measurements. Figure~\ref{fig:figure5} compares the numerical results from Fig.~\hyperref[fig:figure4]{\ref{fig:figure4}a} with those obtained when the loss function is defined by the infidelity instead of the Hilbert-Schmidt distance. The relative performance of the measurements remains unchanged, and no fundamental change in behavior is observed.

As evident from Eq.~\eqref{eq:trace-sep-general}, if the qubits are measured in different Pauli bases, all outcomes are equally likely and therefore independent of the true state. This means that, of the nine Pauli measurements, only three provide nontrivial information. These three measurements are the same as those performed in the Pauli parity check. Consequently, we observe in Fig.~\hyperref[fig:figure4]{\ref{fig:figure4}c} that the average risk for a fixed $N$ for the Pauli measurements is three times as large as for the ordered Pauli parity checks. The Pauli measurements are also significantly worse than the Pauli parity checks for the Bayesian mean estimator in Fig.~\hyperref[fig:figure4]{\ref{fig:figure4}a} and the maximum-likelihood estimator in Fig.~\hyperref[fig:figure4]{\ref{fig:figure4}b}. However, in these cases, the factor is no longer $3$, as the estimates are also influenced by the prior distribution and the constraint that the estimated state must be physical. A similar situation occurs for measurements in the mutually unbiased bases. Three out of the five measurements coincide with the Pauli parity checks, while the other two measurements provide no insight. Hence, the average risk for MUBs measurements of the direct inversion estimator differs by a factor of $5/3$ from the average risk obtained for the Pauli parity checks.

In Figs.~\hyperref[fig:figure4]{\ref{fig:figure4}d}-\hyperref[fig:figure4]{\ref{fig:figure4}g} we compare the different estimators to each other for fixed measurement sets. In all cases, the Bayesian mean estimator produces the lowest average risk. Except for the Bell state measurements, where maximum-likelihood estimation and direct inversion coincide, the average risk of the maximum-likelihood estimator is lower than the average risk of the direct inversion. 

\begin{figure}
  \centering
  \includegraphics[width=\columnwidth]{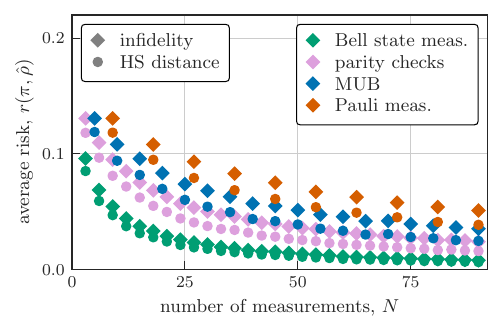}
  \caption{Average risk for Bayesian mean estimation over a uniform prior plotted against the number of measurements $N$ for different types of measurements and different loss functions. The plot displays the numerical results of Fig.~\hyperref[fig:figure4]{\ref{fig:figure4}a} for the Hilbert-Schmidt distance alongside the results obtained when defining the loss function with the infidelity. The same sample size and discretization as in Fig.~\ref{fig:figure4} are used.}
  \label{fig:figure5}
\end{figure}


\section{Discussion}

In summary, we discussed different estimators and measurement sets for Bell diagonal state estimation. We showed that Bell state measurements are optimal in a general setting and introduced the LOCC-implementable Pauli parity checks for Bell diagonal state estimation in a network context. We derived bounds on the measurement statistics of LOCC-implementable measurements on Bell diagonal states and showed that Pauli parity checks meet these bounds. Furthermore, for Pauli parity checks, we provided an analytical expression for the risk of direct inversion and established an upper bound on the average risk of the Bayesian mean estimator. In addition to these analytical results, a numerical analysis comparing different measurement sets and estimators highlights the advantages of Bayesian mean estimation.

The reported improved scaling of the average risk for Bayesian mean estimation, compared to maximum-likelihood estimation or state tomography, is not tied to the specific structure of Bell diagonal states and seen across all discussed measurements. We therefore expect comparable behavior of the estimation risk when analyzing other parametrized states. Beyond the advantages in terms of the average risk, a key motivation for using Bayesian mean estimation over maximum-likelihood estimation or state tomography is its well-defined uncertainty quantification, for example, through the covariance matrix of the posterior distribution, and its natural compatibility with prior knowledge about the state~\cite{blume2010optimal}. Bayesian mean estimation is typically more computationally demanding than alternative methods because it requires representing a probability distribution over a high-dimensional space, a challenge that is often addressed using Monte Carlo sampling techniques~\cite{blume2010optimal, Granade2016, Lukens2020}. However, in the case of Bell diagonal state estimation, where the state can be parametrized by just three parameters, the space is small enough to allow for direct discretization. As a result, Bayesian mean estimation offers substantial advantages for Bell diagonal state estimation without significant additional implementation costs.

In scenarios where not all four components of $\theta$ are of equal interest, it is possible to introduce weights into the loss function of Eq.~\eqref{eq:HS-distance_t}. The corresponding risk can be related to a weighted scalar Cramér-Rao bound (see Sec.~\ref{sec:QCR}), and we believe it is possible to extend the derived average risk expressions for this more general loss function.

While we provided an indication for the conjecture that the Pauli parity checks are the optimal set of LOCC-implementable measurements for Bell diagonal state estimation, no proof is provided in this work. A possible approach could be to minimize the Fisher information over the set of LOCC-implementable POVMs. 

In Sec.~\ref{sec:BellStateMeas} we showed that there is no advantage of collective measurements or adaptive measurements over Bell state measurements. However, in the case of LOCC-implementable measurements, collective measurements on multiple instances are expected to reach lower average risks. An example of this is visible in Fig.~\ref{fig:figure3}, where the ordered Pauli parity checks, a simple example of a POVM acting on three states, achieve lower average risks than the random Pauli parity check. While beyond the scope of this work, there are certainly more elaborate adaptive and collective techniques to explore.

Recent work has explored the possibility of characterizing Bell diagonal states with $\theta_1 > 0.5$ as a by-product of an entanglement distillation protocol~\cite{casapao2024disti}. We leave for future work the analysis of these hybrid estimation protocols with a double figure of merit. 

The presented protocol, like most characterization or tomographic schemes, relies on partial knowledge of the devices and assumes trust in the messages received from the other party. In particular, we assume control over the measurement procedure, and therefore, the protocol is not device independent. In a networked setting, trust between parties typically follows from a user authentication protocol, which is independent of the characterization. Consequently, integrating the proposed protocol into a broader network stack is a feasible extension of this work.

There is no obvious equivalent to the set of Bell states for multidimensional bipartite systems or multipartite qubit systems. For the former case, Ref.~\cite{Sych2009} presented a set of states with many of the properties of the Bell states for bipartite qudit systems where the dimension of each system is a power of 2. That work showed the connections of those states to the generalized Pauli matrices. This could, combined with the insights of this work, serve as a starting point for an estimation protocol for convex combinations of those maximally entangled multidimensional bipartite states. When entanglement is shared among more than two parties, for example, in a quantum network, there is no obvious distinctive measure for multipartite entanglement~\cite{Walter2019} and therefore there exists no unique generalization of bipartite maximally entangled states. For an extension of this work, interesting sets of states to look at would be the absolute maximally entangled states~\cite{Goyeneche2015} and Greenberger–Horne–Zeilinger (GHZ) states. Tripartite GHZ states appear to be particularly interesting, as for those states, each bipartition carries maximum entanglement. Therefore, for each bipartition, the estimation of a convex combination of such states has a similar structure to the Bell diagonal state estimation.

\section*{Acknowledgements}

The project was supported by the JST Moonshot R\&D program under Grant No. JPMJMS226C.

\section*{Data Availability}
The data that support the findings of this article are openly available~\cite{Kaufmann2025}.

\section*{Author Contributions}

N.K. conducted the analysis, developed numerical code, performed computations, generated visualizations, authored the initial manuscript draft, implemented suggested revisions, and prepared the final submission. M.Q. conceptualized the project within the broader research theme of the Networked Quantum Devices unit at OIST, provided regular supervision and guidance on technical derivations and numerical approaches, and contributed detailed feedback on the initial draft. D.E. supervised the project. 

\bibliography{References}

\end{document}